\documentclass[12pt,a4paper]{article}

\usepackage{authblk}
\usepackage{amsmath, amssymb, mathtools, tikz}
\usepackage{amsthm}
\usepackage{color}
\usepackage{hyperref}

\newcommand{\sign}{\mathsf{sign}}
\newcommand{\MAJ}{\mathsf{MAJ}}
\newcommand{\LDL}{\mathsf{LDL}}
\newcommand{\LTF}{\mathsf{LTF}}
\newcommand{\PL}{\mathsf{PL}}
\newcommand{\PT}{\mathsf{PT}}
\newcommand{\OR}{\mathsf{OR}}
\newcommand{\EQ}{\mathsf{EQ}}
\newcommand{\OMB}{\mathsf{OMB}}
\newcommand{\SINK}{\mathsf{SINK}}
\newcommand{\PAR}{\mathsf{PARITY}}

\newcommand{\AND}{\mathsf{AND}}
\newcommand{\XOR}{\mathsf{XOR}}
\newcommand{\R}{\mathbb{R}}
\newcommand{\E}{\mathbb{E}}
\newcommand{\wh}{\widehat}
\newcommand{\bra}[1]{\{#1\}}
\newcommand{\bool}{\{0,1\}}
\newcommand{\boolfn}{\ensuremath {\bool}^n \rightarrow \bool}

\newtheorem{theorem}{Theorem}%[section]
\newtheorem{corollary}[theorem]{Corollary}
\newtheorem{lemma}[theorem]{Lemma}
\newtheorem{definition}[theorem]{Definition}

\newtheorem{fact}[theorem]{Fact}
\newtheorem{remark}[theorem]{Remark}

\begin{document}
\title{Lower bounds for linear decision lists}
\author[1]{Arkadev Chattopadhyay}
\author[2]{Meena Mahajan}
\author[3]{Nikhil Mande\footnote{This work was done while the author was a graduate student at TIFR, Mumbai.}}
\author[4]{Nitin Saurabh}
\affil[1]{Tata Institute of Fundamental Research, Mumbai, India\\
arkadev.c@tifr.res.in}
\affil[2]{The Institute of Mathematical Sciences, HBNI, Chennai, India\\
meena@imsc.res.in}
\affil[3]{Georgetown University, Washington, D.C.~USA\\
nikhil.mande@georgetown.edu}
\affil[4]{Max Planck Institut f\"{u}r Informatik, Saarland Informatics Campus, Saarbr\"ucken, Germany\\
  
nsaurabh@mpi-inf.mpg.de}

\date{}
\maketitle
\begin{abstract}
We demonstrate a lower bound technique for linear decision lists, which are decision lists where the queries are arbitrary linear threshold functions.
We use this technique to prove an explicit lower bound by showing that any linear decision list computing the function $\MAJ \circ \XOR$ requires size $2^{0.18 n}$.  This completely answers an open question of Tur{\'a}n and Vatan \cite{TV97}. We also show that the spectral classes $\PL_1, \PL_\infty$, and the polynomial threshold function classes $\wh\PT_1, \PT_1$,  are incomparable to linear decision lists. 
\end{abstract}

\section{Introduction}

Decision lists are a widely studied model of computation, first introduced by Rivest \cite{Rivest87}.
A decision list $L$ of size $\ell$ computing a Boolean function $f \in B_n$ is a sequence of $\ell - 1$ instructions of the form \textbf{if $f_i(x) = a_i$ then output $b_i$ and stop}, followed by the instruction \textbf{output $\neg b_{\ell-1}$ and stop}. Here $B_n$ denotes the set of all Boolean functions in $n$ variables, each $f_i \in B_n$ is called a \emph{query function}, and $a_i$ and $b_i$ are Boolean constants.  If the functions $f_i$ all belong to a function class $S \subseteq B_n$, then $L$ is said to be an $S$-decision list.

Krause \cite{Krause06} showed that there are functions with small representation as $\AND$-decision lists, but requiring exponential size when computed by depth-two circuits with a linear threshold gate at the top and
$\XOR$ gats at the bottom. 
On the other hand, Impagliazzo and Williams \cite{IW10} showed that a certain condition is sufficient to prove lower bounds against a related computation model that can be termed rectangle-decision lists.
Linear decision lists are decision lists where the query functions are linear threshold functions.  Lower bounds against linear decision lists (and even against bounded-rank linear decision trees, a natural generalisation) for the Inner Product modulo $2$ function were proved by Gr{\"{o}}ger,  Tur{\'{a}}n and Vatan, in  \cite{GT91, TV97}. Subsequently, Uchizawa and Takimoto~\cite{UT11, UT15} showed lower bounds against the class of linear decision lists and linear decision trees when the weights of the linear threshold queries are bounded by a polynomial in the input length.
In fact, the lower bounds of \cite{UT11, UT15} apply to any function with large \emph{unbounded-error communication complexity}.

We observe that the lower bound argument in \cite{TV97} shows that functions efficiently computable by linear decision lists (with no restrictions on the weights of the queried linear threshold functions) must have large monochromatic rectangles.
In fact, we build on their argument to establish a more general result (Lemma~\ref{lem:largerects}).
Informally, we show that if a function has no ``large'' weight monochromatic rectangles under some product distribution then it 
cannot be expressed by ``small'' linear decision lists. 
We then use this fact to establish a lower bound for a seemingly simple function, $\MAJ \circ \XOR$ (see Definition~\ref{def:xor-fns}).
Our main theorem is as follows.

\begin{theorem}\label{thm:main}
  Any linear decision list computing $\MAJ_n \circ \XOR$ must have size
  $2^{\Omega(n)}$.
\end{theorem}

It is not hard to see that $\MAJ \circ \XOR$ can be simulated by
$\MAJ \circ \MAJ$ circuits with only a \emph{linear} blow-up in size.  This immediately yields the following corollary, resolving an open question posed by Tur{\'a}n and Vatan in~\cite{TV97}.
\begin{corollary}\label{cor:main}
There exists a function that can be computed by polynomial sized $\MAJ \circ \MAJ$ circuits, but any linear decision list computing it requires exponential size.
\end{corollary}

Impagliazzo and Williams~\cite{IW10} demonstrated a function, implicitly computable by polynomial sized $\MAJ \circ \MAJ$ circuits, which cannot be computed by polynomial sized rectangle-decision lists. We observe that our lower bound technique against linear decision lists (Lemma~\ref{lem:largerects}) coincides with the sufficient condition considered in~\cite{IW10} to prove lower bounds against rectangle-decision lists.  Thus, their function also separates linear decision lists from $\MAJ \circ \MAJ$. However, we obtain a $2^{\Omega(n)}$ lower bound on the length of linear decision lists in Theorem~\ref{thm:main}, improving upon the bound implicit in the work of Impagliazzo and Williams, which is worse in the exponent by at least a quadratic factor. 
Very recently,
Chattopadhyay, Mande and Sherif~\cite{CMS18} showed several properties of the function $\SINK \circ \XOR$. We observe that as a consequence, our lower bound technique against linear decision lists (Lemma~\ref{lem:largerects}) also applies to this function. 
We elaborate more on these remarks in Section~\ref{sec:hierarchy}.

\section{Preliminaries}

\begin{definition}[Sign function]
  The function $\sign : \R \rightarrow \bool$ is defined as 
  follows. 
\[ \sign(x) = \left\{\begin{array}{ll}
1 & \textrm{~if~} x > 0 \\
0 & \textrm{~if~} x \le 0 
\end{array}
\right.\]
\end{definition}

\begin{definition}[Linear Threshold Functions]
A function $f : \boolfn$ is said to be a linear threshold function ($\LTF$) if there exist real numbers $w_0, w_1, \dots, w_n$ such that
$f(x) = \sign \left(w_0+\sum\limits_{i = 1}^n w_ix_i\right)$.
\end{definition}
For strings $x,y\in \R^n$, we denote their inner product by $\langle
x, y\rangle \triangleq \sum_i x_i y_i$. With this notation, $f$ is an
$\LTF$ if for some $w_0\in \R$, $\tilde{w}\in \R^n$, $f(x) = \sign(w_0 +
\langle \tilde{w},x\rangle)$.

\begin{definition}[Majority]
The function $\MAJ_n : \boolfn $ is the linear threshold function defined by $\MAJ_n(x) = \sign\left(x_1 + x_2 + \cdots + x_n - n/2\right)$.
\end{definition}
\begin{definition}[Function composition]
For functions $f : \boolfn$ and $g : \bool^m \rightarrow \bool$, the function $f \circ g: \bool^{nm} \rightarrow \bool$ is  defined as follows: 
\[
f \circ g(x_{11}, \ldots, x_{1m}, \ldots, x_{n1}, \ldots, x_{nm}) = f(g(x_{11}, \ldots, x_{1m}), \cdots ,g(x_{n1}, \ldots, x_{nm})).
\]
\end{definition}

We now formally define the model of computation that is of interest in this paper.  
\begin{definition}[Linear Decision Lists]
A linear decision list ($\LDL$) of size  $k$ is a sequence $(L_1, a_1), (L_2, a_2), \ldots, (L_k, a_k)$, where each $a_i \in \bool$, 
and each $L_i$ is an $\LTF$ with $L_k$ being the constant function $1$.
The decision list computes a function $f : \boolfn$ as follows : 
If $L_1(x) = 1$, then $f(x) = a_1$; elseif $L_2(x) = 1$, then $f(x) = a_2$; elseif \ldots, elseif $L_k(x) = 1$, then $f(x) = a_k$.  That is,
\[
f(x) = \bigvee_{i = 1}^k\left(a_i \wedge  L_i(x) \wedge \bigwedge_{j < i}\neg L_j(x)\right).
\]
\end{definition}

\begin{definition}[Communication matrix]
  For a function $F : \bool^n \times \bool^n \rightarrow \bool$, its communication matrix $M_F$ is
  the $2^n \times 2^n$ matrix with entries $M_F[x, y] := F(x, y)$.
\end{definition}

\begin{definition}[Monochromatic rectangles/squares]
Let $F : \bra{0, 1}^n \times \bra{0, 1}^n \rightarrow \bra{0, 1}$ be any function.  For $b \in \bool$, a monochromatic $b$-rectangle is a tuple $(X, Y)$, where $X, Y \subseteq \bra{0, 1}^n$
and $F(x,y)=b$ for every $(x,y)\in X \times Y$. 
We say that $(X,Y)$ is a monochromatic square of size $s$ if it is a monochromatic 0-rectangle or 1-rectangle and, furthermore, $|X| = |Y| = s$. 
\end{definition}

\begin{definition}[Product distributions and weights]
A probability distribution $\eta$ over $\bool^n \times \bool^n$ is
said to be a \emph{product distribution} if there are probability
distributions $\mu$, $\nu$ over $\bool^n$ such that for every
$(x,y)\in \bool^n\times \bool^n$, $\eta(x,y) = \mu(x) \times \nu(y)$.
We say that $\eta$ is the product distribution $\mu \times \nu$. 

Given a probability distribution $\mu$ over $\bool^n$ and $X \subseteq
\bool^n$, $\mu(X)$ is defined to be the sum $\sum_{x\in X}\mu(x)$.

For a rectangle $(X,Y)$, its \emph{weight} under a product
distribution $\mu\times \nu$ is $(\mu\times \nu)(X \times Y) = \mu(X)
\times \nu(Y)$.
  
\end{definition}

We will denote the number of $1$'s in a string $x \in \bool^n$ by $|x|$.  

\begin{definition}[Hamming distance]
  The (Hamming) distance between any two strings $x, y \in \bra{0, 1}^n$, denoted $d(x, y)$, is defined as $d(x,y) \triangleq |\bra{i : x_i \neq y_i}|$.
The Hamming distance between any two sets $A, B \subseteq \bool^n$, denoted $d(A,B)$, is the minimum pairwise distance;  $d(A, B) = \min_{x \in A, y \in B}d(x, y)$.
\end{definition}

\begin{definition}[Hamming balls]
  Let $c \in \bool^n$ and $k \in \{1,\ldots , n\}$.
  A set $A \subseteq \bool^n$ is called a Hamming ball with centre $c$
  and radius $k$ if
  \[ \{s \in \bool^n \mid d(s,c) \le k-1 \} \subset A \subseteq
  \{s \in \bool^n \mid d(s,c) \le k \}. \]
  A singleton set $A=\{c\}$ is a Hamming ball with centre $c$ and radius $0$. 
\end{definition}

For a set $A \subseteq \bool^n$, the boundary of $A$ is the set
$\{s\in \bool^n \mid d(s,A) = 1 \}$. In \cite{Harper66}, Harper proved a
isoperimetry result: among all sets of a given size, Hamming balls
have the smallest boundary set size. A simplified proof was given by Frankl
and F\"{u}redi \cite{FF81}, who also stated the theorem in the
equivalent form we mention below. (See also the presentation in
\cite{Bol86}).

\begin{theorem}[Harper's Theorem]\label{thm:harper}
Let $A, B \subseteq \bra{0, 1}^n$ be non-empty sets.  Then, there exists a Hamming ball $A_0$ with centre $0^n$ and a Hamming ball $B_0$ with centre $1^n$ such that $|A_0| = |A|$, $|B_0| = B$, and $d(A_0, B_0) \geq d(A, B)$.
\end{theorem}

\begin{definition}[Binary Entropy]
The binary entropy function  $\mathbb{H} : [0, 1] \rightarrow [0, 1]$ is defined as follows: $\mathbb{H}(p) = -p\log p - (1-p)\log (1-p)$.
\end{definition}
\begin{fact}\label{fact:binom}
\(
\mathbb{H}(1/4) < 0.82
\). 
\end{fact}

\section{Linear decision lists contain large monochromatic rectangles}

The argument of Tur{\'a}n and Vatan from \cite{TV97} implicitly showed that
any function $f : \bool^n \times \bool^n \to \bool$ with no large monochromatic squares cannot be computed by small linear decision lists. Their argument was presented specific to the Inner Product function (Theorem~1 in \cite{TV97}).
However, it is not too hard to see
that their proof in fact works for any function as long as it has no large
monochromatic squares. 
In this section, we generalize their argument to show that all functions computable by small size linear decision lists must contain,
under \emph{any product distribution},
a  monochromatic rectangle  of large weight with respect to that distribution.

We first establish a technical lemma that can be seen as a generalization of Lemma 2 in \cite{TV97}.

\begin{lemma}\label{lem:norects}
  Let $f$ be an $\LTF$ over the input variables $x_1, \ldots, x_n, y_1, \ldots, y_n$. Let $\mu, \nu$ be distributions over $\bool^n$, and $X, Y \subseteq \bool^n$. Define $m := \min\bra{\mu(X), \nu(Y)}$, and let $t \in (0, m]$. Then,
one of the following is true.
\begin{enumerate}
 \item There exists a monochromatic 1-rectangle $(X', Y')$ within $X \times Y$ (i.e., $X' \subseteq X$ and $Y' \subseteq Y$) such that $\mu(X') \geq t$ and $\nu(Y') \geq t$.
 \item There exists a monochromatic 0-rectangle $(X', Y')$ within $X \times Y$ such that $\mu(X') > m - t$ and $\nu(Y') > m -t$.
\end{enumerate}
\end{lemma}
\begin{proof}
Let $M$ be the submatrix of $M_f$ restricted to $X \times Y$.  Let the $\LTF~f$ be given by $\sign(a + \langle \alpha \cdot x\rangle + \langle \beta \cdot y\rangle) $.  
Reorder the rows and columns of $M$ in decreasing order of $a+\langle \alpha \cdot x\rangle$ and $\langle \beta \cdot y\rangle$ to get the matrix $B = R \times C$.
Let $i$ denote the least index of a row in $B$ such that $\mu(\bra{R_1, \dots R_i}) \geq t$, and $j$ denote the least index of a column in $B$ such that $\mu(\bra{C_1, \dots C_j}) \geq t$.
Note that these indices are well-defined since $t \in (0, m]$.
If the $[i, j]$'th entry of $B$ is 1, then the top-left submatrix of $B$ satisfies item~(1) in the lemma.
If the $[i, j]$'th entry of $B$ is 0, then the bottom-right submatrix of $B$ satisfies item~(2) in the lemma.
\end{proof}

We now prove the main lemma. 

\begin{lemma}\label{lem:largerects}
Let $\mu, \nu$ be distributions on $\bool^n$. Let $f : \bool^n \times \bool^n \to \bool$ be any function with no monochromatic rectangle of weight greater than $w$ under the distribution $\mu \times \nu$.
Then, any linear decision list computing $f$ must have size at least $1/\sqrt{w}$.
\end{lemma}
\begin{proof}
  Towards a contradiction, let $(L_1, a_1), (L_2, a_2), \ldots, (L_k, a_k)$ be an $\LDL$ of size $k$ computing $f$, where $k < 1/\sqrt{w}$. Pick any $t \in (\sqrt{w},1/k]$.  We construct, for each $i \in [k - 1]$, a rectangle $S_i = X_i \times Y_i$ which is a $0$-rectangle for all $L_j$ with $j \leq i$, and furthermore $\mu(X_i), \nu(Y_i) \geq 1 - i\cdot t$.  
We proceed by induction on $i$. 

  For the base case $i = 1$, let $S_0 =  (X_0,Y_0)$ be the entire $2^n \times 2^n$ matrix. Suppose $S_0$ has a rectangle $(X',Y')$ that is a $1$-rectangle of $L_1$ and moreover, $\mu(X')\geq t$, $\nu(Y') \geq t$. Then everywhere in this rectangle, $f$ will be $a_1$. But $f$ has no monochromatic rectangle of weight as large as $t^2 > w$. So $S_0$ has no rectangle $(X',Y')$ with $\mu(X') \geq t$, $\nu(Y') \geq t$ that is a $1$-rectangle of $L_1$.  By Lemma~\ref{lem:norects}, $S_0$ must then contain a $0$-rectangle $(X_1,Y_1)$ of $L_1$  such that both $\mu(X_1)$ and $\nu(Y_1)$ are at least $1-t$. This establishes the base case.

For the inductive step, we have a rectangle $S_{i-1} = (X_{i-1}, Y_{i-1})$ which is a $0$-rectangle for $L_1, L_2, \ldots ,L_{i-1}$ and, moreover, $\min\{\mu(X_{i-1}),\nu(Y_{i-1})\} \geq 1 - (i-1)t$. Within this rectangle, suppose $L_i$ has a $1$-rectangle $(X',Y')$ such that $\mu(X')\geq t$ and $\nu(Y') \geq t$. 
Then $f = a_i$ in this rectangle, giving a monochromatic rectangle of $f$ of weight greater than $w$. But we know that such rectangles do not exist.
Since $kt\le 1$ and $i < k$, we have $t \le 1-(i-1)t$ and hence Lemma~\ref{lem:norects} is applicable. Hence we conclude that $S_{i-1}$ must contain a $0$-rectangle $(X_i,Y_i)$ of $L_i$ with $\min\{\mu(X_i),\nu(Y_i)\} \geq 1 - (i-1)t -t = 1 -it$. Since this rectangle, say $S_i$, is contained in $S_{i-1}$, it is a $0$-rectangle for all $L_j$ with $j \leq i$.

Thus, we have a rectangle $S_{k-1} = (X_{k-1},Y_{k-1})$ on which $L_1, L_2, \ldots , L_{k-1}$ are $0$, and $L_k=1$ because $L_k$ is the constant function $1$. Furthermore, $\mu(X_{k-1})$ and  $\nu(Y_{k-1}) \geq 1 - (k-1)t$.  Everywhere on this rectangle, $f$ evaluates to $a_k$. So $S_{k-1}$ is a monochromatic rectangle for $f$. Hence it cannot have weight more than $w$. Thus $1-(k-1)t \le \sqrt{w} < t$; that is, $1 < kt$, contradicting our choice of $t$.
\end{proof}

\section{$\MAJ \circ \XOR$ has no large monochromatic squares}
In this section, we show an upper bound and a matching tight lower bound on the size of a largest monochromatic square in the communication matrix of the $\MAJ \circ \XOR$ function.
\begin{definition}[$\XOR$ functions]
  \label{def:xor-fns}
For a function $f : \boolfn$, let $f \circ \XOR$ denote the function defined by $f \circ \XOR(x_1, \ldots, x_n, y_1, \ldots y_n) = f(x_1 \oplus y_1, \ldots, x_n \oplus y_n)$.
\end{definition}

\begin{lemma}
  Let $F : \bool^n \times \bool^n \rightarrow \bool$ be the function $\MAJ_n \circ \XOR$. Then, for any $b \in \bool$, $M_F$ has a monochromatic $b$-square of size at least $\sum\limits_{i = 0}^{\lfloor n/4 \rfloor}\binom{n}{i}$.
\end{lemma}
\begin{proof}
  Define the sets $X,Y,Z$ as follows:  
  \begin{eqnarray*}
    X = Y &=& \bra{x \in \bool^n : |x| \leq \lfloor n/4 \rfloor}. \\
  Z & = & \bra{x \in \bool^n : |x| \geq n - \lfloor n/4 \rfloor}.
  \end{eqnarray*}
  Note that $F(x, y) = 0$ for all $x \in X, y \in Y$, and $F(x, z) =
  1$ for all $x \in X$, $z \in Z$. Thus $(X,Y)$ and $(X,Z)$ are a
  monochromatic 0-square and 1-square, respectively, each of size
  $\sum\limits_{i = 0}^{\lfloor n/4 \rfloor}\binom{n}{i}$.
\end{proof}

\begin{remark}
  \label{rem:monochromatic-recgtangle-Maj-Xor}
  We remark that when $n \equiv 3 \pmod{4}$ the above construction can be improved if we consider monochromatic rectangles. That is, for any $b \in \bool$, $M_F$ has a monochromatic $b$-rectangle $(X_1,X_2)$ such that $|X_1| = \sum\limits_{i = 0}^{\lceil n/4 \rceil}\binom{n}{i}$ and $|X_2| = \sum\limits_{i = 0}^{\lfloor n/4 \rfloor}\binom{n}{i}$. Indeed, let $X = \bra{x \in \bool^n : |x| \leq \lceil n/4 \rceil}$, $Y =\{x \in \bool^n : |x| \leq \lfloor n/4 \rfloor\}$  and $Z = \bra{x \in \bool^n : |x| \geq n - \lfloor n/4 \rfloor}$. Then, it is easily seen
  that $(X,Z)$ (resp., $(X,Y)$) is a monochromatic $1$-rectangle (resp., $0$-rectangle) of the claimed size. 
\end{remark}

We now show that this bound is tight.

\begin{theorem}\label{thm:nolargesquares}
Let $F : \bool^n \times \bool^n \rightarrow \bool$ be the function $\MAJ_n \circ \XOR$.  For any $n$, $M_F$ has no monochromatic squares of size greater than $\sum\limits_{i = 0}^{\lceil n/4 \rceil}\binom{n}{i}$.
\end{theorem}
\begin{proof}
Suppose, to the contrary, that there are sets $A, B \subseteq \bool^n$ such that $|A| = |B| > \sum\limits_{i = 0}^{\lceil n/4\rceil}\binom{n}{i}$ and $A \times B$ is a monochromatic $1$-square in $M_F$. By the definition of $F$, this implies $d(A, B) > \lfloor n/2 \rfloor$.
By Theorem \ref{thm:harper}, there exist Hamming balls $A_0$ around $0^n$, and $B_0$ around $1^n$ such that $|A_0| = |A|, |B_0| = |B|$ and $d(A_0, B_0) \geq d(A, B)$.
The size lower bound enforces that the radius of $A_0$ and $B_0$ must be greater than $\lceil n/4 \rceil$, and since they are centered on
$0^n$ and $1^n$, it follows that $d(A_0,B_0) \leq \lfloor n/2 \rfloor$.
But then $ d(A,B)$ is also at most $\lfloor n/2 \rfloor$.
Hence, there exist $x \in A, y \in B$ such that $d(x, y) \leq \lfloor n/2 \rfloor$, which means $F(x,y) = \MAJ_n \circ \XOR(x, y) = 0$, which contradicts our assumption. 
Therefore, any monochromatic $1$-square in
$M_F$ has size at most $\sum\limits_{i = 0}^{\lceil n/4 \rceil}\binom{n}{i}$.

A similar argument
shows the same upper bound on the size of monochromatic $0$-squares. 
\end{proof}

Now we can put things together to prove our main theorem. 
\begin{proof}[Proof of Theorem \ref{thm:main}]
  Let $s_n$ be the minimum size of an $\LDL$ computing $\MAJ_n \circ \XOR$.
  Further let $\mu$ and $\nu$ be uniform distributions over $\bool^n$.
  Then, by Lemma \ref{lem:largerects} and Theorem \ref{thm:nolargesquares},
  for all $n$ sufficiently large, 
\begin{align*}
s_n & \geq \frac{2^n}{\sum_{i = 0}^{\lceil n/4 \rceil}\binom{n}{i}}\\
& \geq \frac{2^n}{2^{n \cdot H(1/4) }}\tag*{using Stirling's approximation}\\
& \geq 2^{0.18 n}. \tag*{using Fact \ref{fact:binom}}
\end{align*}
\end{proof}

\section{$\LDL$s and the threshold circuit hierarchy}\label{sec:hierarchy}

In this section, we see how the class of functions computable by polynomial sized $\LDL$s fits into the low depth threshold circuit hierarchy.  
The reader is referred to Razborov's survey \cite{Raz92} for a detailed exposition on the low depth threshold circuits hierarchy.

\subsection{Definitions}

\begin{definition}[$\MAJ$]
Define $\MAJ$  to be the class of all functions computable by polynomial sized $\MAJ$  gates. Each input to the $\MAJ$ gate may be a constant 0 or 1, or a variable $x_i$, or its negation $\neg x_i$. 
\end{definition}
    
\begin{definition}[$\LTF$]
Define $\LTF$ to be the class of all functions computable by 
$\LTF$ gates.
\end{definition}

\begin{definition}[$\LDL$]
Define $\LDL$ to be the class of all functions computable by polynomial sized linear decision lists.
\end{definition}

\begin{definition}[$\wh\LDL$]
Define $\wh\LDL$ to be the class of all functions computable by
polynomial sized linear decision lists where, furthermore, weights of
the linear threshold queries are integers with values bounded by a
polynomial in the number of variables.
\end{definition}

\begin{definition}[Depth-2 classes]
For classes of functions $\mathcal{C}, \mathcal{D}$, define $\mathcal{C} \circ \mathcal{D}$ to be the class of functions computable by polynomial-sized depth-2 circuits, where the top  gate computes a function from the class $\mathcal{C}$, and the bottom layer contains gates computing functions in $\mathcal{D}$.
\end{definition}

\begin{definition}[$\wh{\PT_1}$]
The class $\wh{\PT_1}$ consists of all functions $f : \boolfn$ which can be represented by polynomial sized $\MAJ \circ \PAR$ circuits.
\end{definition}

\begin{definition}[$\PT_1$]
The class $\PT_1$ consists of all functions $f : \boolfn$ which can be represented by polynomial sized $\LTF \circ \PAR$ circuits.
\end{definition}

(These are precisely the classes of polynomial threshold functions \cite{Bruck90}; it is
more convenient for us here to use the equivalent formulation as
depth-2 circuits.)

In order to define classes given by the spectral representation of functions, we first recall a few preliminaries from Boolean function analysis.

Consider the real vector space of functions from $\bool^n \rightarrow \R$, equipped with the following inner product.
\[
\langle f, g \rangle = \frac{1}{2^n}\sum\limits_{x \in \bool^n}f(x)g(x) = \E_{x \in \bool^n}[f(x)g(x)].
\]
For each $S \subseteq [n]$, define $\chi_S : \bool^n \rightarrow \bra{-1, 1}$ by $\chi_S(x) = (-1)^{\sum_{i \in S}x_i}$.
It is not hard to verify that $\bra{\chi_S : S \subseteq [n]}$ forms an orthonormal basis for this vector space.
Thus, every $f: \bool^n \rightarrow \R$ has a unique representation as $f= \sum\limits_{S \subseteq [n]}\wh{f}(S)\chi_S$, where
\begin{align*}
\wh{f}(S) = \langle f, \chi_S\rangle = \E_{x \in \bool^n}[f(x)\chi_S(x)].
\end{align*}

\begin{definition}[$\PL_1$]
The class $\PL_1$ consists of all functions $f : \boolfn$ for which $\sum\limits_{S \subseteq [n]}|\wh{f}(S)| \leq \mathrm{poly}(n)$.
\end{definition}

\begin{definition}[$\PL_\infty$]
The class $\PL_\infty$ consists of all functions $f : \boolfn$ for which $\max\limits_{S \subseteq [n]}|\wh{f}(S)| \geq \frac{1}{\mathrm{poly}(n)}$.
\end{definition}

\begin{figure}
\begin{center}
\begin{tikzpicture}[transform shape]

\node (pl1) at (0,0) {$\PL_1$};
\node (pt1hat) at (0, 2) {$\widehat{\PT_1}$};
\node (pt1) at (0,4) {$\PT_1$};
\node (plinfinity) at (0,6) {$\PL_\infty$};
\node (maj) at (6, 0) {$\MAJ$};
\node (thr) at (6, 2) {$\LTF$};
\node (majomaj) at (6, 4) {$\MAJ \circ \MAJ$};
\node (thromaj) at (6, 6) {$\LTF \circ \MAJ$};
\node (throthr) at (6, 8) {$\LTF \circ \LTF$};
\node (ldl) at (12, 4) {$\LDL$};
\node (mdl) at (12,2) {$\wh{\LDL}$};

\draw[->,thick] (pl1)--(pt1hat);
\draw[->,thick] (pt1hat)--(pt1);
\draw[->,thick] (pt1)--(plinfinity);
\draw[->,thick] (maj)--(thr);
\draw[->,thick] (thr)--(majomaj);
\draw[->,thick] (majomaj)--(thromaj);
\draw[->,thick] (thromaj)--(throthr);

\draw[->,thick] (maj)--(pt1hat);
\draw[->,thick] (thr)--(pt1);
\draw[->,thick] (pt1hat)--(majomaj);
\draw[->,thick] (pt1)--(thromaj);

\draw[->] (mdl)--(ldl);
\draw[->,thick] (thr)--(ldl);
\draw[->,thick] (ldl)--(throthr);
\draw[->,thick] (maj)--(mdl);
\draw[->,thick] (mdl)--(thromaj);

\path[color=red, dashed] (mdl) edge  (majomaj);
\path[color=red, dashed] (ldl) edge  (majomaj);
\path[color=red, bend left = 6, dashed] (ldl) edge (pt1hat);
\path[color=red, bend left = 10, dashed] (plinfinity) edge (mdl);
\path[color=red, bend left = 6, dashed] (ldl) edge (pl1);

%\path[->, color=red, dashed] (majomaj) edge (ldl);
%\path[bend left = 6, dashed] (thromaj) edge (ldl);
%\path[<-, bend left = 6] (throthr) edge (ldl);

\end{tikzpicture}
\caption{Low depth threshold circuit hierarchy}
\label{fig: fig}
\end{center}
\end{figure}
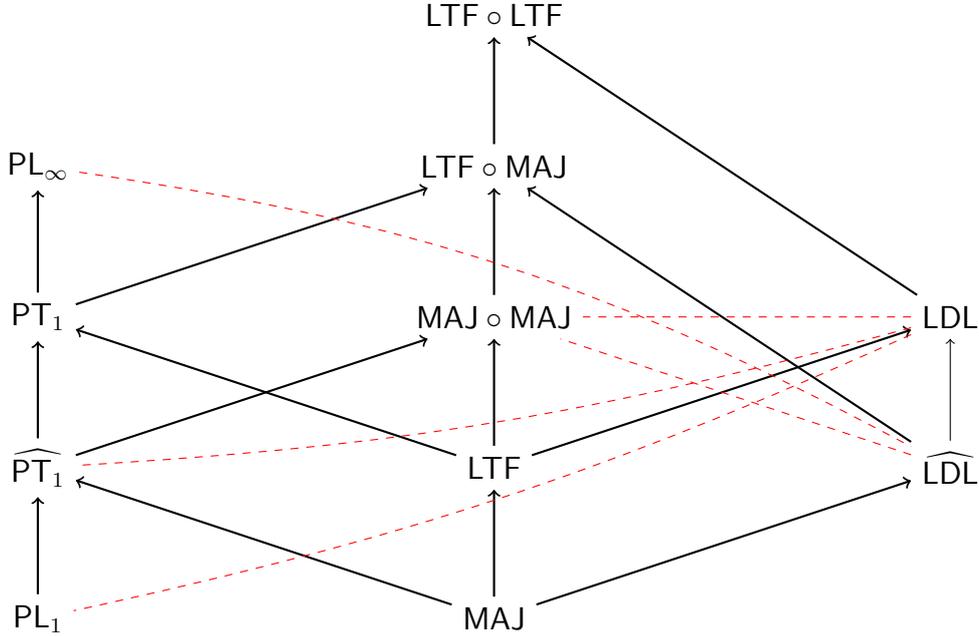

Figure \ref{fig: fig} depicts the currently known status of low depth circuit class containments, and shows where linear decision lists fit in this hierarchy.

A thick solid arrow from $\mathcal{C}_1$ to $\mathcal{C}_2$ denotes
$\mathcal{C}_1 \subsetneq \mathcal{C}_2$, a thin solid arrow from
$\mathcal{C}_1$ to $\mathcal{C}_2$ denotes $\mathcal{C}_1 \subseteq
\mathcal{C}_2$, and a dashed line between $\mathcal{C}_1$ and
$\mathcal{C}_2$ denotes incomparability. In the figure, we only show
the newly established incomparabilities.

The leftmost column has the classes defined based on spectral
representation, and the middle column has the classes based on depth-2
circuits. Concerning these classes, the picture was already completely clear:
All containments shown among classes in these columns are known to be
strict, and wherever no arrow connects two classes, they are known to
be incomparable. Essentially this part of the figure appears in
\cite{GHR92}; a subsequent refinement is the insertion of the class
$\LTF\circ\MAJ$, separated from $\MAJ\circ\MAJ$ in \cite{GHR92}, from
$\PT_1$ in \cite{Bruck90} and most recently from $\LTF\circ \LTF$ in
\cite{CM18}.

The two classes $\wh\LDL$ and $\LDL$ form the new column on the right.
In the following subsection we explain their position with respect to
the other two  columns. However here the picture is not yet
completely clear, and there are still several open questions. 

\subsection{New results}
By definition, $\MAJ \subseteq \wh\LDL$ and $\LTF \subseteq \LDL$ via
lists of size 2.  The parity function is known to not be in $\LTF$,
and it has a simple $\LDL$ with 0-1 weights in the query functions. Thus
both these containments are proper, and $\wh\LDL$ is not contained in $\LTF$.
We now observe that, implicit from prior work,
$\wh\LDL$ is not even contained in $\MAJ \circ \MAJ$.

\begin{theorem}\label{thm:omb}
\[
\wh\LDL \nsubseteq \MAJ \circ \MAJ.
\]
\end{theorem}
\begin{proof}  
Define the ODD-MAX-BIT function by $\OMB(x) = 1$ iff the largest index $i$ where $x_i = 1$ is odd ($\OMB(0^n) = 0$).
Buhrman, Vereshchagin and de~Wolf~\cite{BVW07} showed that $\OMB \circ \AND$ is hard, in the sense that it has exponentially small \emph{discrepancy}.  By a result of Hajnal, Maass, Pudl\'ak, Szegedy and Tur\'an~\cite{HMPST93},
this implies that $\OMB \circ \AND$ cannot be computed by polynomial sized $\MAJ \circ \MAJ$ circuits.

Note that $\OMB$ can be computed by a linear sized decision list by querying the variables in decreasing order of their indices. Thus $\OMB \circ \AND$ can be computed by a linear sized decision list of $\AND$'s, and hence by a linear decision list with 0-1 weights.
\end{proof}

On the other hand, it is easily seen that $\MAJ_n \circ \XOR$ is in 
$\MAJ \circ \MAJ$, and even in $\wh\PT_1$ (see for instance~\cite{Bruck90}).
Combining this with Theorem~\ref{thm:main}, we obtain:
\begin{theorem}\label{thm:whpt1-not-in-ldl}
\[
\wh\PT_1 \nsubseteq \LDL .
\]
\end{theorem}

The following strengthening of Theorem~\ref{thm:whpt1-not-in-ldl} is implicit from a recent result of Chattopadhyay, Mande and Sherif~\cite{CMS18}.
\begin{theorem}
  \label{thm:pl1-not-in-ldl}
  \[\PL_1 \not\subseteq \LDL.\]
\end{theorem}
(We defer a discussion of why Theorem~\ref{thm:pl1-not-in-ldl} holds to Section~\ref{sec:pl1-not-in-ldl}.)
Putting together these separations with the known containments $\PL_1 \subseteq \wh\PT_1 \subseteq \MAJ\circ\MAJ$, we obtain a slew of incomparability results. 
\begin{corollary}\label{cor:incomp1}
For any class 
$A \in \{\wh\LDL, \LDL\}$ and $B \in \{\PL_1, \MAJ\circ \MAJ\}$,
the classes $A$ and $B$ are incomparable.
\end{corollary}

In particular, the classes $\LDL$ and $\MAJ\circ\MAJ$ are incomparable.
This completely answers the
open question posed by Tur{\'a}n and Vatan \cite{TV97}.

Impagliazzo and Williams~\cite[Theorem~4.8]{IW10} showed that the function $\OR_n \circ \EQ_n$
(also called  Block-Equality) does not contain large monochromatic rectangles (in fact they showed that it does not contain large monochromatic rectangles under any product distribution). Thus, by Lemma~\ref{lem:largerects},
any linear decision list computing $\OR_n \circ \EQ_n$ must be of size at least $2^{\Omega(n)}$. 
We now observe that $\OR \circ \EQ \in \MAJ \circ \MAJ$.
Consequently, $\OR \circ \EQ$ also witnesses $\MAJ \circ \MAJ \nsubseteq \LDL$.
However, in contrast to Theorem~\ref{thm:main}, note that the lower bound is subexponential since $\OR \circ \EQ$ is defined on $2n^2$ variables. Moreover,
$\OR\circ \EQ$ seems to incur a significant polynomial blow up in size when simulated by $\MAJ \circ \MAJ$ circuits, whereas $\MAJ_n \circ \XOR$ has linear sized $\MAJ \circ \MAJ$ circuits.  
\begin{theorem}
\[
\OR \circ \EQ \in \MAJ \circ \MAJ.
\]
\end{theorem}
\begin{proof}
First observe that $\OR \circ \EQ$ can be computed by a $\MAJ \circ \EQ$ circuit by suitably padding constants to the input.
Next, note that $\EQ$ is an \emph{exact threshold function}, that is there exist reals $a_1, \ldots, a_n, b_1, \ldots, b_n, c$ such that $\EQ(x,y) = 1$ iff $\sum_{i = 1}^n a_ix_i + b_iy_i = c$.
Hansen and Podolskii~\cite{HP10} showed that such functions can be efficiently simulated by $\MAJ \circ \LTF$ circuits.
However, we do not need the full strength of their result, so we
give a direct construction below. 

For an equality on $2n$ bits, say $x_1, \ldots, x_n, y_1, \ldots, y_n$, note that
\[
\EQ_n(x_1, \ldots, x_n, y_1, \ldots, y_n) = 1 \iff \sum_{i=1}^n 2^i(x_i - y_i) = 0.
\]
Consider the following linear threshold functions.
\begin{align*}
g_1(x, y) & = \sign\left(\sum_{i=1}^n 2^i(x_i - y_i) + 1/2)\right) \text{~and}\\
g_2(x, y) & = \sign\left(\sum_{i=1}^n 2^i(x_i - y_i) - 1/2)\right).
\end{align*}
Observe that $g_1(x, y) - g_2(x, y) = \EQ_n(x, y)$.

Let $g_1^{(i)}$ and $g_2^{(i)}$ denote these $\LTF$s for the $i$th block
on which we test equality.  The function $\OR_n\circ\EQ_n$ is just
\[\OR_n\circ\EQ_n = 
\sign\left((g_1^{(1)}-g_2^{(1)}) + (g_1^{(2)}-g_2^{(2)}) + \ldots + (g_1^{(n)}-g_2^{(n)})\right);\] this formulation puts it in $\MAJ \circ \LTF$.  

Finally, Goldmann, H{\aa}stad and Razborov~\cite{GHR92} showed that $\MAJ \circ \LTF = \MAJ \circ \MAJ$. 
Thus, $\OR \circ \EQ \in \MAJ \circ \MAJ$.
\end{proof}

\begin{theorem}\label{thm:whldl-not-in-plinfty}
\[
\wh\LDL \nsubseteq \PL_\infty.
\]
\end{theorem}
\begin{proof}
It is easy to see that any symmetric function (a function that only depends on the Hamming weight of the input) can be computed by linear sized linear decision lists where query functions are majority:
the linear threshold queries can be used to determine the Hamming weight of the input, and the decision list outputs the appropriate answer at each decision.

Bruck~\cite{Bruck90} showed that the \emph{Complete Quadratic} function, which is a symmetric function, is not in $\PL_\infty$.
This function yields the required separation.
\end{proof}

Combining Corollary~\ref{cor:incomp1} and Theorem~\ref{thm:whldl-not-in-plinfty} yields more incomparability results.
\begin{corollary}\label{cor:incomp2}
For any class 
$A \in \{\wh\LDL, \LDL\}$ and $B \in \{\PL_1,\PL_\infty\}$, 
the classes $A$ and $B$ are incomparable. In other words,
all spectral classes in the first column (see Figure~\ref{fig: fig}) are incomparable to all classes in the third column. 
\end{corollary}

Finally, as noted in \cite{TV97}, $\LDL$ is contained in $\LTF \circ \LTF$. 
The same argument shows that $\wh\LDL$ is contained in $\LTF \circ \MAJ$. 
Corollary~\ref{cor:incomp1} implies that these containments are strict.

\subsection{Proving Theorem~\ref{thm:pl1-not-in-ldl}}
\label{sec:pl1-not-in-ldl}
As mentioned earlier, it is implicit from a recent result of Chattopadhyay et al.~\cite{CMS18} that $\PL_1 \nsubseteq \LDL$.
We first define the function used to achieve the separation and introduce some background required.

\begin{definition}[$\SINK$]
Consider a complete undirected graph on $n$ vertices with variables $x_{i, j}$ for $i < j \in [n]$.  The variable $x_{i,j}$ assigns a direction to the edge between $v_i$ and $v_j$ in the following way:  $x_{i, j} = 0$ implies the edge points towards $v_i$, and $x_{i, j} = 1$ implies the edge points towards $v_j$.  The function $\SINK$ computes whether or not there is a sink in the graph.  In other words,
\[
\SINK(x) = 1 \iff~ \exists i \in [n]~\text{such that all edges adjacent to}~i~\text{are incoming}.
\]
\end{definition}
We now define the notion of \emph{projections} of strings to certain subsets of coordinates.
Let $X \in \bool^{\binom{n}{2}}$.  For any vertex $v_i$, let $E_{v_i}$ be the set of $n - 1$ coordinates corresponding to the $n - 1$ edges adjacent to $v_i$.  Let $X_{v_i}$ denote the $(n - 1)$-bit string obtained by projecting $X$ to the coordinates in $E_{v_i}$.

\begin{definition}[Entropy]
Let $X$ be a discrete random variable. The entropy $H(X)$ is defined as
\[
H(X) = \sum_{s \in \textnormal{supp(X)}} \Pr[X=s] \log \frac{1}{\Pr[X=s]}.
\]
\end{definition}

\begin{fact}[Folklore]
\label{fact:entropy-support}
$\textnormal{supp}(X) = k \implies H(X) \leq \log k$, with equality if and only if $X$ is uniform.
\end{fact}

\begin{lemma}[Shearer's Lemma~\cite{CGFS86} (see also \cite{Jaikumar2001})]\label{lem:shearer}
Let $X = (X_1, \dots, X_t)$ be a random variable.  If $S$ is a set of projections such that for each $i \in [t]$,~$i$ appears in at least $k$ projections, then $\sum_{P \in S}[H_{X_P}] \geq kH(X)$.
\end{lemma}

Chattopadhyay et al.~\cite{CMS18} introduced and used the function $\SINK \circ \XOR$ to refute the long-standing Log-Approximate-Rank Conjecture, along with several other conjectures.  
They observe that $\SINK \circ \XOR \in \PL_1$~\cite[Theorem~1.10]{CMS18}.
\begin{lemma}[Part 1 of Theorem~1.10 in~\cite{CMS18}]
\[
\SINK \circ \XOR \in \PL_1.
\]
\end{lemma}

It is also implicit from their work that $\SINK \circ \XOR$ does not contain large monochromatic rectangles under the uniform distribution.  More precisely, plugging the value $\epsilon = 0$ in \cite[Claim~6.4]{CMS18} implies that any monochromatic rectangle in the communication matrix of $\SINK \circ \XOR$ on $2\binom{n}{2}$ variables must have weight at most $2^{2\binom{n}{2} - \Omega(n)}$.  However, we do not require the full power of their proof for our purpose, and therefore produce a self-contained proof below.

\begin{theorem}\label{thm:sink-xor-rectangle}
Any monochromatic rectangle $R = A \times B$ in the communication matrix of $\SINK \circ \XOR$ must satisfy $|R| \leq 2^{2\binom{n}{2} - n + \log n +1}$.
\end{theorem}

\begin{proof}
  It is easy to verify that the probability of a 1-input under the uniform distribution equals $n/2^{n - 1}$.  Hence if  $R$ is a 1-monochromatic rectangle, then
$|R| \le 2^{2\binom{n}{2}} \times n / 2^{n-1}$, as claimed in the theorem.

  Let $R = A \times B$ be a 0-monochromatic rectangle. Consider the random variable $XY$ ($X$ concatenated with $Y$) over $2\binom{n}{2}$ coordinates, when $X$ and $Y$ are sampled uniformly from $A$ and $B$, respectively. From Fact~\ref{fact:entropy-support} we have $H(XY) = \log |R|$.
  
  Let $S$ be the set of projections $S := \{E_{v_i} \mid 1 \leq i \leq n\}$. Then each coordinate appears in exactly two projections. Hence by Lemma~\ref{lem:shearer},
  \[2H(XY) \le \sum_{P\in S} H((XY)_P) = \sum_{i\in [n]} H((XY)_{v_i}). \]
  We  now bound the entropy in $XY$ restricted to each of the projections. Let $A_{v_i}$ and $B_{v_i}$ be the projections of $A$ and $B$ on $E_{v_i}$, respectively. Since there is no input in $R$ which is a sink, we have $|\textnormal{supp}(A_{v_i})| + |\textnormal{supp}(B_{v_i})| \leq 2^{n - 1}$. (Each string in $A_{v_i}$ rules out one string from $B_{v_i}$ and vice versa.) By the AM-GM inequality,  $|\textnormal{supp}(A_{v_i})| \cdot |\textnormal{supp}(B_{v_i})| \leq 2^{2n - 4}$.  Hence Fact~\ref{fact:entropy-support} implies that  $H((XY)_{v_i}) \leq 2n - 4$.

  Returning to our use of Lemma~\ref{lem:shearer}, we obtain 
\begin{align*}
& 2H(XY) \leq \sum_{P \in S}H((XY)_P) \leq n(2n - 4)\\
\implies & H(XY) \leq 2\binom{n}{2} - n\\
\implies & |R| \leq 2^{2\binom{n}{2} - n}. 
\end{align*}
\end{proof}

Along with Lemma~\ref{lem:largerects}, Theorem~\ref{thm:sink-xor-rectangle} shows that any linear decision list computing the function $\SINK \circ \XOR$ on $2\binom{n}{2}$ variables (which is in $\PL_1$) must have size at least $2^{n/2}$.
This completes the proof of Theorem~\ref{thm:pl1-not-in-ldl}. 

  Clearly, $\SINK \circ \XOR$ also witnesses $\MAJ \circ \MAJ \not\subseteq \LDL$. However, the lower bound against $\LDL$ is again only subexponential.  

\section{Conclusions}

We show that $\MAJ \circ \XOR$ cannot be computed by polynomial sized linear decision lists, resolving an open question of Tur{\'a}n and Vatan \cite{TV97}. We also show that several spectral classes and polynomial threshold function classes are incomparable to linear decision lists. 
Figure~\ref{fig: fig} depicts where the class $\LDL$, and its small-weight version $\wh\LDL$, fit in the low depth threshold circuit hierarchy.

A subset of the authors~\cite{CM18} showed that a decision list of \emph{exact threshold functions} cannot be computed by $\LTF \circ \MAJ$. 
A natural question that arises is whether $\LDL$ is incomparable with $\LTF \circ \MAJ$. 
(Note that the function from \cite{CM18} separating $\LTF\circ \LTF$ from $\LTF \circ \MAJ$ does not settle this question as it is also not in $\LDL$ --  it contains the function $\OR\circ \EQ$ as a subfunction.)

Another natural question is whether
$\wh\LDL$ is strictly contained in $\LDL$; that is, whether weights matter in linear decision lists. 

\section{Acknowledgments}
We thank Rahul Santhanam for discussions concerning decision lists. 
We thank Jaikumar Radhakrishnan for referring us to Harper's theorem.

\bibliography{bibo}

\begin{thebibliography}{10}

\bibitem{Bol86}
B{\'e}la Bollob{\'a}s.
\newblock {\em Combinatorics: set systems, hypergraphs, families of vectors,
  and combinatorial probability}.
\newblock Cambridge University Press, 1986.

\bibitem{Bruck90}
Jehoshua Bruck.
\newblock Harmonic analysis of polynomial threshold functions.
\newblock {\em {SIAM} J. Discrete Math.}, 3(2):168--177, 1990.

\bibitem{BVW07}
Harry Buhrman, Nikolay Vereshchagin, and Ronald de~Wolf.
\newblock On computation and communication with small bias.
\newblock In {\em Proceedings of the Twenty-Second Annual IEEE Conference on
  Computational Complexity}, CCC '07, pages 24--32. IEEE Computer Society,
  2007.

\bibitem{CM18}
Arkadev Chattopadhyay and Nikhil~S. Mande.
\newblock A short list of equalities induces large sign rank.
\newblock In {\em Proc.\ 59th Annual IEEE Symposium on Foundations of Computer
  Science (FOCS)}, 2018.
\newblock Preliminary version in ECCC TR 2017-083.

\bibitem{CMS18}
Arkadev Chattopadhyay, Nikhil~S. Mande, and Suhail Sherif.
\newblock The log-approximate-rank conjecture is false.
\newblock {\em Electronic Colloquium on Computational Complexity {(ECCC)}},
  25:176, 2018.

\bibitem{CGFS86}
Fan R.~K. Chung, Ronald~L. Graham, Peter Frankl, and James~B. Shearer.
\newblock Some intersection theorems for ordered sets and graphs.
\newblock {\em Journal of Combinatorial Theory, Ser. {A}}, 43(1):23--37, 1986.

\bibitem{FF81}
Peter Frankl and Zolt{\'{a}}n F{\"{u}}redi.
\newblock A short proof for a theorem of {H}arper about {H}amming-spheres.
\newblock {\em Discrete Mathematics}, 34(3):311--313, 1981.

\bibitem{GHR92}
Mikael Goldmann, Johan H{\aa}stad, and Alexander~A. Razborov.
\newblock Majority gates {vs.~} general weighted threshold gates.
\newblock {\em Computational Complexity}, 2:277--300, 1992.

\bibitem{GT91}
Hans~Dietmar Gr{\"{o}}ger and Gy{\"{o}}rgy Tur{\'{a}}n.
\newblock On linear decision trees computing boolean functions.
\newblock In {\em Automata, Languages and Programming, 18th International
  Colloquium, ICALP91, Madrid, Spain, July 8-12, 1991, Proceedings}, pages
  707--718, 1991.

\bibitem{HMPST93}
Andr{\'{a}}s Hajnal, Wolfgang Maass, Pavel Pudl{\'{a}}k, Mario Szegedy, and
  Gy{\"{o}}rgy Tur{\'{a}}n.
\newblock Threshold circuits of bounded depth.
\newblock {\em J. Comput. Syst. Sci.}, 46(2):129--154, 1993.

\bibitem{HP10}
Kristoffer~Arnsfelt Hansen and Vladimir~V. Podolskii.
\newblock Exact threshold circuits.
\newblock In {\em Proceedings of the 25th Annual {IEEE} Conference on
  Computational Complexity, {CCC} 2010, Cambridge, Massachusetts, June 9-12,
  2010}, pages 270--279, 2010.

\bibitem{Harper66}
L.~H. Harper.
\newblock Optimal numberings and isoperimetric problems on graphs.
\newblock {\em Journal of Combinatorial Theory}, 1:385--393, 1966.

\bibitem{IW10}
Russell Impagliazzo and Ryan Williams.
\newblock Communication complexity with synchronized clocks.
\newblock In {\em Proceedings of the 25th Annual {IEEE} Conference on
  Computational Complexity, {CCC} 2010, Cambridge, Massachusetts, June 9-12,
  2010}, pages 259--269, 2010.

\bibitem{Krause06}
Matthias Krause.
\newblock On the computational power of boolean decision lists.
\newblock {\em Computational Complexity}, 14(4):362--375, 2006.

\bibitem{Jaikumar2001}
Jaikumar Radhakrishnan.
\newblock Entropy and counting.
\newblock In J.C. Misra, editor, {\em Computational Mathematics, Modelling and
  Algorithms}. Narosa Publishers, New Delhi, 2001.
\newblock IIT Kharagpur Golden Jubilee Volume.

\bibitem{Raz92}
Alexander~A. Razborov.
\newblock On small depth threshold circuits.
\newblock In {\em Third Scandinavian Workshop on Algorithm Theory (SWAT)},
  pages 42--52, 1992.

\bibitem{Rivest87}
Ronald~L. Rivest.
\newblock Learning decision lists.
\newblock {\em Machine Learning}, 2(3):229--246, 1987.

\bibitem{TV97}
Gy{\"o}rgy Tur{\'a}n and Farrokh Vatan.
\newblock Linear decision lists and partitioning algorithms for the
  construction of neural networks.
\newblock In {\em Foundations of Computational Mathematics}, pages 414--423.
  Springer, 1997.

\bibitem{UT11}
Kei Uchizawa and Eiji Takimoto.
\newblock Lower bounds for linear decision trees via an energy complexity
  argument.
\newblock In {\em Mathematical Foundations of Computer Science 2011 - 36th
  International Symposium, {MFCS} 2011, Warsaw, Poland, August 22-26, 2011.
  Proceedings}, pages 568--579, 2011.

\bibitem{UT15}
Kei Uchizawa and Eiji Takimoto.
\newblock {\em Lower Bounds for Linear Decision Trees with Bounded Weights},
  pages 412--422.
\newblock Springer Berlin Heidelberg, Berlin, Heidelberg, 2015.

\end{thebibliography}

\end{document}